\documentclass[journal,letterpaper]{IEEEtran}
\usepackage{graphicx}
\usepackage{epstopdf}
\usepackage{amsmath,amssymb,amsthm,mathrsfs,amsfonts,dsfont}
\usepackage{epsfig}
\usepackage[ruled,linesnumbered]{algorithm2e}
\usepackage{color}
\usepackage{subfigure}
\usepackage{cite}
\usepackage{diagbox}
\usepackage{multirow}

\newcommand{\tabincell}[2]{\begin{tabular}{@{}#1@{}}#2\end{tabular}}

\newtheorem{theorem}{Theorem}

\newtheorem{lemma}{Lemma}

\DeclareMathOperator*{\maximize}{maximize}
\DeclareMathOperator*{\subj}{subject\,to}

\SetKwInput{KwData}{Parameters}

\begin{document}
\title{On Energy Efficient Uplink Multi-User \\MIMO with Shared LNA Control}
\author{Zehao~Yu,
	~Cong~Shen,
        ~Pengkai~Zhao, 
        and~Xiliang~Luo
\thanks{Z. Yu and C. Shen are with the Laboratory of Future Networks, School of Information Science and Technology, University of Science and Technology of China. P. Zhao is with Qualcomm, Santa Clara, CA, USA. X. Luo is with School of Information Science and Technology, ShanghaiTech University.}        
}
\maketitle

\begin{abstract}
Implementation cost and power consumption are two important considerations in modern wireless communications, particularly in large-scale multi-antenna systems where the number of individual radio-frequency (RF) chains may be significantly larger than before. In this work, we propose to deploy a single low-noise amplifier (LNA) on the uplink multiple-input-multiple-output (MIMO) receiver to cover all antennas. This architecture, although favorable from the perspective of cost and power consumption, introduces challenges in the LNA gain control and user transmit power control. We formulate an energy efficiency maximization problem under practical system constraints, and prove that it is a constrained quasi-concave optimization problem. We then propose an efficient algorithm, \emph{Bisection -- Gradient Assisted Interior Point (B-GAIP)}, that solves this optimization problem. The optimality, convergence and complexity of B-GAIP are analyzed, and further corroborated via numerical simulations. In particular, the performance loss due to using a shared LNA as opposed to separate LNAs in each RF chain, when using B-GAIP to determine the LNA gain and user transmit power, is very small in both centralized and distributed MIMO systems. 

\end{abstract}

\section{Introduction}
\label{sec:intro}

Energy efficiency of communication systems is of significant practical importance and has become a hot research topic in both academia and industry. This is mainly due to the increasing global energy demand and the requirement of green radio \cite{chen2011fundamental}. In addition, despite the significant development of battery technology, it has not been able to fully keep pace with the practical demand from portable devices such as smartphones and tablets \cite{lahiri2002battery}. From the operators' perspective, reducing both the operation cost and carbon dioxide emissions \cite{akbari2010energy} is  becoming essential to their business bottom line. As a result, research on energy efficient wireless communications has been prolific over the past decade \cite{onireti2013energy,Shi2016Energy,jang2014energy,Geraci2015Energy}.

The energy efficient design become especially crucial with the introduction of multiple-input-multiple-output (MIMO), particularly with the increased emphasis on \emph{massive} MIMO \cite{Marzetta:2010} in 5G standards \cite{Andrews:2014}. A typical system architecture of massive MIMO assigns a separate radio-frequency (RF) chain to each transmit or receive antenna. When the number of antennas is large, hardware cost and power consumption increase substantially, which has motivated extensive studies on massive MIMO with inexpensive hardware components such as low-resolution Analog-to-Digital Converter (ADC) \cite{Mollen:2017}, Digital-to-Analog Converter (DAC) \cite{Li:2017}, mixers and oscillators \cite{Bjornson:2014}. 

In this paper, we follow the same design philosophy and study an attractive low-complexity MIMO receiver structure, where a single low-noise amplifier (LNA) \cite{Zhao2017Energy} is used to cover all receive antennas at the base station (BS). This architecture has the benefits of reduced implementation cost and lower power consumption, compared to the separate LNA approach where each RF chain uses an independent LNA for gain control. The shared-LNA structure is previously used by multi-channel communications \cite{Zhao2017Energy} where signals from different channels are \emph{non-overlapping} in the frequency domain. This feature mostly relies on LNA's wider bandwidth and more relaxed saturation point compared to other RF components like ADC. Fortunately, we will show in this work that the shared-LNA structure can be adopted by (large-scale) multi-antennas even when they are using the same spectrum. One intuitive solution is to program down-conversion parameters of different receiver paths (especially the configuration of mixer and by-pass filters), so that within shared LNA, signals from these receiver paths are not overlapping in the frequency domain. 

However, reducing the number of LNAs also introduces some important design challenges. For the separate LNA receiver structure, each receive antenna will have an independent LNA to adjust the power of the received signal for further processing. This gain can be optimized based on the individual receive power of the RF chain, resulting in maximum flexibility. For the shared LNA structure, however, the single LNA gain control must accommodate all receive antennas. Hence, it is conceivable that performance degradation may occur if inappropriate power amplification happens on some RF chains, resulting in ADC overflow or underflow\footnote{In our work, similar to the existing literature, we also consider using low-resolution ADCs \cite{singh2009limits, mezghani2007ultra, mo2015capacity} following the LNA gain control, to reduce the power consumption and overall cost.}.

Intuitively, the performance disadvantage of shared LNA may be significant when the range of receive power values across all BS antennas is large, and the channel paths experience independent fading (both large- and small-scale). In this scenario, a single LNA may not satisfy the power amplification requirements for all antennas, resulting in a performance degradation. To evaluate this interesting case, we study both \emph{centralized} and \emph{distributed} MIMO layouts in this paper, and focus on the large-scale regime (i.e., massive MIMO). For the centralized layout, \cite{Marzetta:2010} studies the system capacity where the number of BS antennas approaches infinity while the number of single-antenna users remains fixed. In addition, a more complete asymptotic analysis where the number of BS antennas and/or users approaches infinity has been carried out in \cite{Guthy:2013}. Regarding the energy efficiency analysis, most of the studies are carried out addressing different aspects such as power allocation algorithm \cite{Zhao2013Energy}, transmit antenna selection \cite{li2013energy}, and link adaptation \cite{Miao2010Energy}. In the {distributed} MIMO layout, the BS antennas are spread out in the coverage area and connected to the BS via fiber or cable \cite{Ngo:2015}. This architecture has recently attracted a lot of research interest because of its potential in offering higher data rate \cite{Dai:2015}, owing to the reduced minimum access distance of users to the scattered BS antennas. The capacity of multi-user large-scale MIMO systems with distributed layout has been evaluated  in \cite{Dai:2015,Ngo:2015,Hosseini:2014}, but studies on its energy efficiency are limited \cite{chen2012energy}.

In this paper, we study the energy efficient system design of an uplink multi-user MIMO (MU-MIMO) system deploying the shared LNA receiver structure, with both centralized and distributed MIMO layouts. More specifically, we focus on the joint optimization of shared LNA gain control and user transmit power control that can optimize the system energy efficiency, which is defined as the ratio between  spectral efficiency and overall energy consumption \cite{kwon1986channel}. We first formulate the energy efficiency optimization problem under realistic engineering constraints, and then show that it is a constrained quasi-concave optimization problem. An efficient algorithm, \emph{Bisection -- Gradient Assisted Interior Point (B-GAIP)}, is proposed and its optimality is proved.  Furthermore, we analyze its convergence and complexity with the help of an equivalent interpretation of B-GAIP. Numerical simulation results are provided to evaluate the benefits of shared LNA.  The main contributions of this paper are summarized as follows.
\begin{itemize}
\item We propose a shared LNA receive structure for uplink MU-MIMO systems, which has reduced implementation/operation cost and near-optimal energy efficiency.
\item We formulate the energy efficiency optimization problem by considering several practical constraints. To solve this problem, we transfer the original problem under a fixed LNA gain into a constrained quasi-concave optimization problem, and then prove its concavity  with respect to the LNA gain. These properties guarantee the feasibility and accuracy of our proposed solution.
\item We propose B-GAIP, which is a two-step algorithm that finds the optimal power vector and LNA gain. By using the combination of gradient assisted interior point and bisection search, the algorithm solves the  energy efficiency optimization problem in an efficient manner. We also show that the algorithm guarantees convergence to the global optimal solution and analyze its complexity.
\end{itemize}

The rest of this paper is organized as follows. Section \ref{sec:model} presents the system model. Section \ref{sec:problem} introduces the constraints and formulates the optimization problem. In Section \ref{sec:algorithm} we design and evaluate the proposed algorithm. Section \ref{sec:simulation} presents  comprehensive numerical simulations to evaluate the performance. Finally, conclusions are drawn in Section \ref{sec:conclusions}.

\textit{Notations:} Throughout this paper, vectors are written as bold letters $\mathbf{x}$, and can be either row or column and their dimensions will be explicitly stated when defined. Matrices are written as bold capital letters $\mathbf{A}$.
$\mathbf{x} \circ \mathbf{y}$ represents the Hadamard product of two vectors, and $\mathbf{A}^{H}$ denotes the Hermitian of $\mathbf{A}$.
$|| \cdot ||$ denotes the $l_2$ norm unless stated otherwise, and $|\mathcal{X}|$ denotes the cardinality of set $\mathcal{X}$.
$[\mathbf{A}]_{ij}$ is the element at the $i$th row and $j$th column of the matrix $\mathbf{A}$.
$x \sim \mathcal{CN}(\bar{x}, \sigma^2)$ denotes a complex Gaussian random variable $x$ with mean $\bar{x}$ and variance $\sigma^2$.
diag($x_1$ , \ldots , $x_n$) denotes an $n\times{}n$ diagonal matrix with diagonal elements $x_1$ , \ldots , $x_n$.

\section{System Model}
\label{sec:model}

Consider an uplink single-cell MU-MIMO system with radius $R_0$. For the convenience of analysis, we assume a circular coverage area centered around the BS. In the system, $K$ user equipments (UEs) are randomly and uniformly distributed in the coverage area, and each UE is equipped with a single antenna. The BS deploys $M$ antennas, which may locate either entirely at the cell center (centralized MIMO) or  randomly and uniformly distributed in the coverage area and connect to the BS via fiber or cable (distributed MIMO) \cite{Ngo:2015}. We denote the set of all BS antennas as $\mathcal{M}$ and the set of UEs as $\mathcal{K}$, with cardinalities $|\mathcal{M}| = M$ and $|\mathcal{K}| = K$, respectively. Note that in both layouts, all the UEs are randomly and uniformly distributed over the cell. Also, in both centralized and distributed MIMO layouts, signals from all antennas will be jointly processed.


Assume that all UEs simultaneously transmit data to the base station, the received vector at the BS can be written as
\begin{equation}
\label{eqn:received signal}
\mathbf{y}=\mathbf{GPx}+\mathbf{z},
\end{equation}
where $\mathbf{y} \in \mathbb{C}^{M \times 1}$ is the signal vector at the BS receive antennas and $\mathbf{z} \sim \mathcal{CN}(\mathbf{0},\sigma_N^2\mathbf{I}_M) \in \mathbb{C}^{M \times 1}$ is an additive white Gaussian noise (AWGN) vector with mean $\mathbf{0}$ and covariance $\sigma_N^2\mathbf{I}_M$, with $\mathbf{I}_M$ denoting the identity matrix with dimension $M$. $\mathbf{P}=$ diag($\sqrt{p_1}$ , \ldots , $\sqrt{p_K}$) is the real-valued diagonal transmit amplitude matrix, and $\mathbf{P}\mathbf{x} \in \mathbb{C}^{K \times 1}$ is the transmitted vector of the $K$ UEs.  $\mathbf{G}\in \mathbb{C}^{M \times K}$ is the channel matrix between $K$ UEs and $M$ BS antennas, whose elements is $g_{mk}\triangleq{}[\mathbf{G}]_{mk}$. The channel matrix $\mathbf{G}$ models the independent fast fading, geometric attenuation, and log-normal shadow fading. As a result, element $g_{mk}$ is given by
\begin{equation}
g_{mk}=h_{mk}\sqrt{\beta_{mk}},
\end{equation}
where $h_{mk}$ is the fast fading coefficient from the $k$th UE to the $m$th BS antenna, and it follows a circularly symmetric complex Gaussian distribution with zero mean and unit variance; $\sqrt{\beta_{mk}}$ represents the geometric attenuation and shadow fading which are assumed to be independent and constant over the coherent intervals. 

In this paper, we adopt the WINNER II path loss model in \cite{kyosti2007winner}, where the path loss in dB domain is
\begin{equation}
\label{eqn:path loss}
\beta_{mk}^{\text{dB}}=46+20\log{}_{10}(d_{mk})+V_{mk}.
\end{equation}
In model \eqref{eqn:path loss}, $d_{mk}$ is the distance from UE $k$ to BS antenna $m$ and $V_{mk}$ denotes the shadow fading which follows the log-normal distribution. The power decay can be written as $\beta_{mk}=10^{(-\beta_{mk}^{\text{dB}}/10)}$.

In order to reduce implementation cost and conserving energy,  we use a single LNA to amplify the received signals from all $K$ UEs at the BS, as opposed to the traditional one-LNA-per-RF-chain approach. The receiver structure is illustrated in Fig.~\ref{fig:LNA}, where a \emph{common} LNA is applied to amplify the signals of \emph{all} receive antennas. The gain of this common LNA is denoted as $\Omega^\text{dB}$ in the dB domain and $\Omega=10^{(\Omega^{\text{dB}}/10)}$. The amplified received signal vector can be written as
\begin{equation}
\tilde{\mathbf{y}}=\sqrt{\Omega}\mathbf{y}.
\end{equation}
We consider a finite range with discrete values for parameter $\Omega^\text{dB}$, i.e., $\Omega_\text{min}^\text{dB} \leqslant{} \Omega^\text{dB} \leqslant{} \Omega_\text{max}^\text{dB}$ and $\Omega^\text{dB}$ is an integer.

\begin{figure}
\centering
\includegraphics[width=0.45\textwidth]{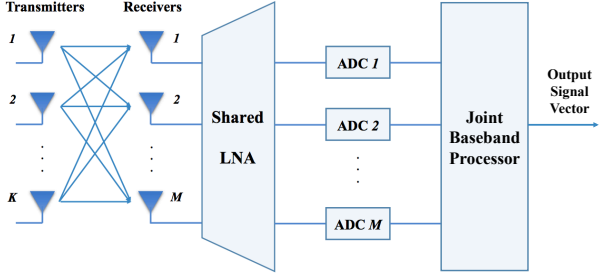}
\caption{The receiver structure with shared LNA control. Note that transmitters in the above illustration are from different uplink users. The receive antennas are for one BS and can be either co-located or distributed.}
\label{fig:LNA}
\end{figure}

After power amplification, each component of the signal vector will pass through an individual low-resolution ADC. We adopt the fixed ADC noise model\footnote{Note that under most of the ADC models, such as the additive quantization noise model (AQNM), the power of quantization noise changes with the power of ADC input signals. However, since LNA is used to control the power gain, it is convenient and appropriate to assume a fixed ADC noise \cite{Bertocco:2000,Zhao2017Energy,zhang2016mixed}.} as in \cite{Bertocco:2000}:
\begin{equation}
\hat{\mathbf{y}}=\tilde{\mathbf{y}}+\mathbf{n}_q,
\end{equation}
where the additive noise vector $\mathbf{n}_q\in \mathbb{C}^{M \times 1}$ is uncorrelated with the ADC input $\tilde{\mathbf{y}}$, and its elements are modeled as independent complex Gaussian random variables with zero mean and variance $\sigma_{\text{ADC}}^2$.

We assume that the BS has perfect knowledge of the CSI. By using a zero-forcing (ZF) detector $\mathbf{F} \triangleq{} (\mathbf{G}^{H}\mathbf{G})^{-1}\mathbf{G}^{H}$ (which requires $M \geqslant K$), the quantized signal vector $\hat{\mathbf{y}}$ is processed as follows:
\begin{equation}
\mathbf{r}=\mathbf{F}\hat{\mathbf{y}} = (\mathbf{G}^{H}\mathbf{G})^{-1}\mathbf{G}^{H} \hat{\mathbf{y}}.
\end{equation}
Since we have $\mathbf{F}\mathbf{G}=\mathbf{I}_K$, $\mathbf{r}$ is  given by
\begin{equation}
\mathbf{r}=\sqrt{\Omega}\mathbf{Px}+\sqrt{\Omega}\mathbf{F}\mathbf{z}+\mathbf{F}\mathbf{n}_q.
\end{equation}
Take the $k$th component of the vector $\mathbf{r}$ as an example, we have
\begin{equation}
\label{eqn:component}
r_k=\sqrt{\Omega{}p_k}x_k+\sqrt{\Omega}\mathbf{f}_k{}\mathbf{z}+\mathbf{f}_k\mathbf{n}_q,
\end{equation}
where $\mathbf{f}_k$ denotes the $k$th row of matrix $\mathbf{F}$. As a result, the signal-to-noise ratio (SNR) of the $k$th UE at the output of the BS receiver can be calculated as
\begin{equation}
\Gamma_{k}=\frac{\Omega{}p_{k}}{(\Omega\sigma_N^2+\sigma_{\text{ADC}}^2)\|\mathbf{f}_k\|^2}.
\end{equation}

By using SNR $\Gamma_{k}$, we define the spectral efficiency (SE) via modified Shannon capacity:
\begin{equation}
\label{eqn:se-def}
R_{k}=\left\{ \begin{array}{ll}
\log_2(1+A_d*\Gamma_{k}), & \Gamma_{k}<\Gamma_{\text{max}}\\
\log_2(1+A_d*\Gamma_{\text{max}}), & \Gamma_{k}\geqslant\Gamma_{\text{max}}\\
\end{array} \right.
\end{equation}
where $A_d$ denotes the coding gain and possibly multi-antenna diversity gain, which in practice is obtained via off-line fitting via link adaptation simulations. $\Gamma_{\text{max}}$ is the maximum achievable SNR at the receiver, which is often dominated by phase noise and IQ mismatch\footnote{IQ mismatch refers to phase and gain imbalance between in-phase (I) and quadrature (Q) paths \cite{fouladifard2003frequency}. For a given $\Gamma_{\text{max}}$ from RF impairment, the baseband demodulation capability is accordingly designed
with no extra demodulation gain when the input SNR is beyond $\Gamma_{\text{max}}$. Such observation motivates the usage of a SNR cap in Eqn.~\eqref{eqn:se-def}.}.

Finally, the energy efficiency is defined as the ratio between spectral efficiency and consumed power of the system \cite{kwon1986channel}. Note that in our model, and also in other literature \cite{Zhao2017Energy,Miao2010Energy}, the overall consumed power includes the circuit power and the transmit power. Therefore, the energy efficiency defined here is a system-level metric rather than that of only the tranceivers. As a result, we have:
\begin{equation}
\label{eqn:ee-def}
U(\mathbf{p}, \Omega)=\frac{\sum_{k=1}^{K}R_{k}}{P_c+\sum_{k=1}^{K}p_k/\eta},
\end{equation}
where $P_c$ denotes the circuit power of both the transmitters and the receivers, $\eta$ is the power amplifier efficiency, and $\mathbf{p}=[p_1,p_2,...,p_K]$ is the power allocation vector. Moreover, we define SE vector under configuration $\mathbf{p}$ and $\Omega$ as $\mathbf{R}=[R_{1},R_{2},...,R_{K}]$. Note that there is a one-to-one mapping between power $p_k$ and spectral efficiency $R_{k}$ for a given $\Omega$. Hence, $U(\mathbf{p},\Omega)$ can also be written as $U(\mathbf{R},\Omega)$. We further use $U(\Omega)$ to denote the maximum energy efficiency under all feasible power vectors.

\subsection{Implementation Considerations for Shared Power Amplifier}
\label{sec:imp_sLNA}


\begin{figure}
	\centering
	\subfigure[ Super-heterodyne implementation ]{ \includegraphics[width=\linewidth]{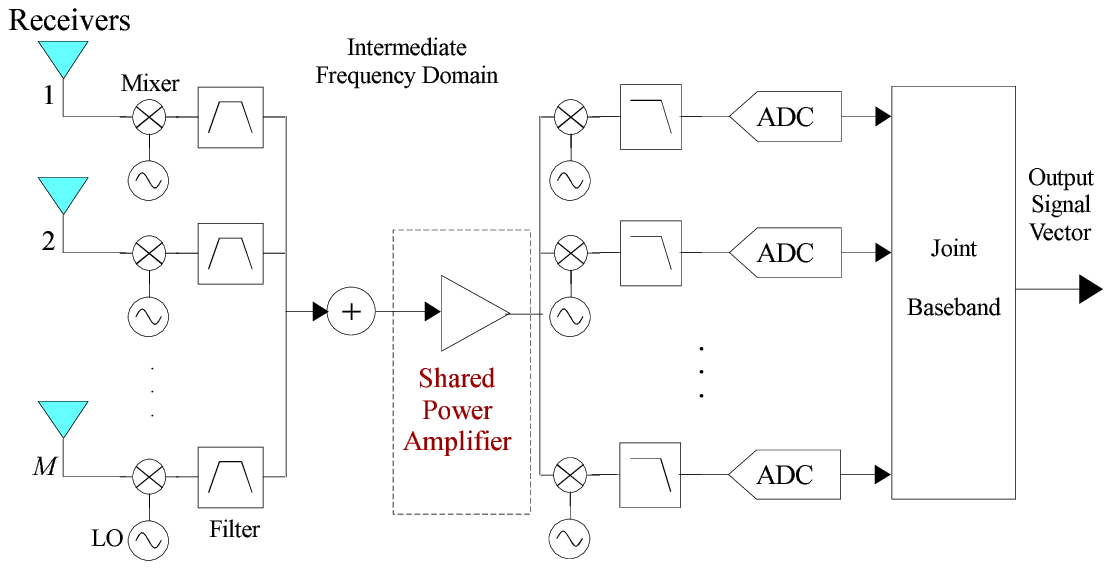}\label{fig:impLNA1}}
	\subfigure[ High-speed switch implementation ]{ \includegraphics[width=\linewidth]{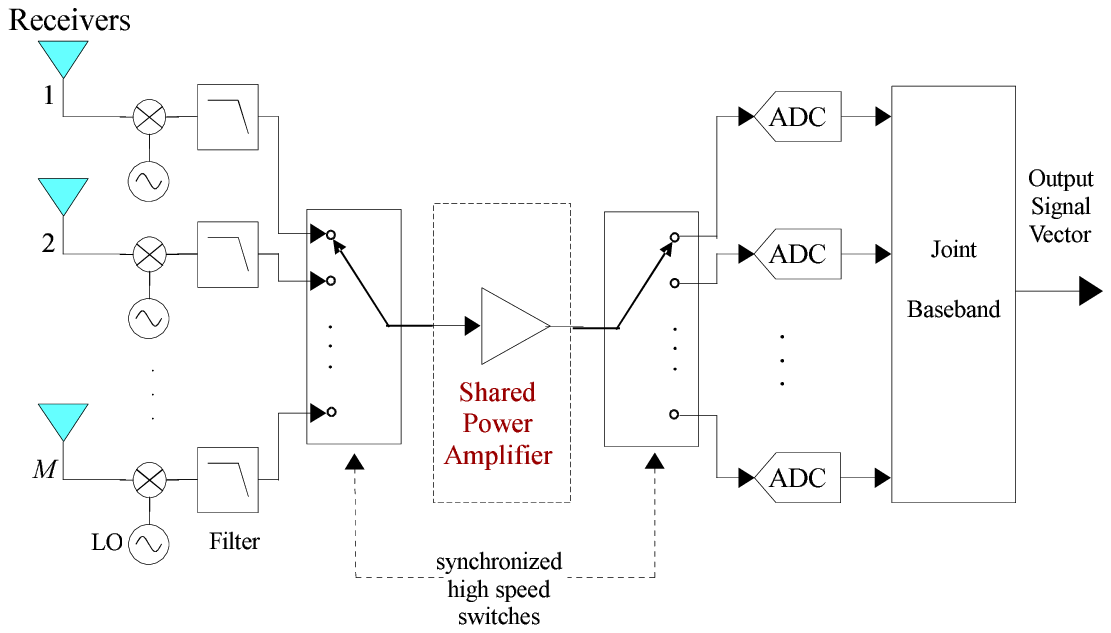}\label{fig:impLNA2}}
	\caption{Two possible shared amplifier implementation structures.}
	\label{fig:impLNA}
\end{figure}

In Fig.~\ref{fig:LNA}, we have illustrated the concept of a shared LNA that uses one power amplifier for all receive RF chains in an uplink MU-MIMO system. Conventionally, shared LNA has been adopted in multi-channel communications \cite{Zhao2017Energy} where signals from different channels are non-overlapping in the frequency domain. It is not straightforward how such a shared LNA structure can be extended to multi-antenna receivers, where signals from different antennas are on the same frequency. In this section, we discuss two possible implementations, shown in Fig.~\ref{fig:impLNA}, that can effectively and efficiently implement shared power amplifier for a multi-antenna receiver.

The first implementation is to leverage the RF framework proposed in \cite{shippee2007shared}, which is built on a super-heterodyne receiver. The details of this implementation is depicted in Fig.~\ref{fig:impLNA1}. In this structure, although different receive antennas are using the same frequency channel, by programming the frond-end mixers and filters, these receive (Rx) paths can be tuned as non-overlapping in the frequency domain at the input of the shared amplifier. In addition, mixers and filters after the amplifier can further isolate the shared-amplifier output from each baseband path.

Another possible implementation is shown in Fig.~\ref{fig:impLNA2}, where a direct-conversion receiver is used and synchronized switches at the input and output of shared amplifier are applied. Assuming sample-and-hold type ADCs and power amplifiers, if the switches at the input and output of shared amplifier have well synchronized timing to ensure the same Rx path, and switching periodicity is aligned with the ADC sampling rate, then the baseband processor can collect the right digital samples from different Rx paths using one  amplifier.

A final comment is that although we focus on LNA in this paper, our work can be directly extended to any power amplifiers in the Rx path of the receiver, such as Intermediate Frequency (IF) amplifier or baseband amplifier.  Besides these two possible implementations, other novel RF structures for shared power amplifiers can be further developed, which is an open topic in the 5G RF research.

\section{Formulating the Energy efficiency Optimization Problem}
\label{sec:problem}
Before presenting the system design problem and looking into its structure, we first introduce the constraints that capture three important engineering limitations in practical systems.


\begin{itemize}

\item Each UE's transmit power $p_k$ is subject to a maximum power value, and is obviously non-negative:
\begin{equation}
\label{limitation1}
0\leqslant{}p_{k}\leqslant{}P_{\text{max}},\ k\in\mathcal{K}.
\end{equation}

\item To avoid ADC saturation, each ADC's input power is capped by a maximum value $P_{\text{max}}^{\text{\text{ADC}}}$:
\begin{equation}
\label{limitation2}
\Omega(\mathbf{g}_m\mathbf{P}^2\mathbf{g}_m^H+\sigma_N^2)\leqslant{}P_{\text{max}}^{\text{\text{ADC}}},\ m\in\mathcal{M},
\end{equation}
where $\mathbf{g}_m$ denotes the $m$th row of matrix $\mathbf{G}$. Note that for each $m\in\mathcal{M}$, there exits a combined limitation on all of the transmit power $P_k,\ \forall k\in\mathcal{K}$.

\item Since the effective SNR and SE at the receiver are capped, the limitation on transmit power can be presented as
\begin{equation}
\label{limitation3}
\frac{\Omega{}p_{k}}{(\Omega\sigma_N^2+\sigma_{\text{ADC}}^2)\|\mathbf{f}_k\|^2}\leqslant{}\Gamma_{\text{max}},\ k\in\mathcal{K}.
\end{equation}
Note that \eqref{eqn:se-def} has already shown that, when the SNR at the receiver achieves the maximum value $\Gamma_{\text{max}}$, the spectral efficiency will stay at $\log_2(1+A_d*\Gamma_{\text{max}})$ without further increase. This means that when the received and processed signal of the $k$th UE already achieved the maximal SNR, there is no gain to increase transmit power $p_k$. As a result, the limitation in \eqref{limitation3} is an equivalent interpretation of \eqref{eqn:se-def}.

\end{itemize}

We comment that these constraints will limit the tunable parameters into a bounded subspace. With these practical limitations,  the energy efficiency maximization problem under a shared LNA can be formally presented as
\begin{align}
\label{eqn:orig_opt_prb}
\notag & \maximize_{\mathbf{p},\,\Omega}\ U(\mathbf{p}, \Omega)\\
& \subj \ \eqref{limitation1},\,\eqref{limitation2},\,\eqref{limitation3}.
\end{align}

We first note that there are two optimizable variables $\mathbf{p}$ and $\Omega$. For $\Omega$, since only a finite set of values can be used, we always have
\begin{equation}
\label{eqn:seq}
\maximize_{\mathbf{p},\,\Omega}\ U(\mathbf{p}, \Omega)=\maximize_{\Omega}\ \maximize_{\mathbf{p}}\ U(\mathbf{p}, \Omega),
\end{equation}
which means that we can optimize $\mathbf{p}$ and $\Omega$ sequentially.

The following lemma shows that, when the LNA gain $\Omega$ is fixed, the aforementioned three  constraints  form a convex set with respect to the power vector $\mathbf{p}$.

\begin{lemma}
\label{lemma:constr}
Under a fixed LNA gain $\Omega$, \eqref{limitation1}, \eqref{limitation2} and \eqref{limitation3} are all linear constraints on the power vector $\mathbf{p}$, and therefore form a convex set with respect to $\mathbf{p}$.
\end{lemma}

\begin{proof}
\eqref{limitation1} and \eqref{limitation3} are straightforward convex sets. For \eqref{limitation2}, it can be re-written as
\begin{equation}
\label{limitation2'}
\Omega\left(\sum_{k=1}^K|g_{mk}|^2p_k+\sigma_N^2\right)\leqslant{}P_{\text{max}}^{\text{ADC}},\ m\in\mathcal{M},
\end{equation}
which is a combined linear limitation on all the components of the power vector, and thus is convex. Finally, the lemma is immediately proven using the property that convexity is preserved under set intersections \cite{boyd2004convex}.

\end{proof}

We proceed to evaluate the objective function $U(\mathbf{p}, \Omega)$ in \eqref{eqn:ee-def}.  Theorem~\ref{th:quasi-concave} states that $U(\mathbf{p}, \Omega)$ is a strictly quasi-concave function under a fixed LNA gain $\Omega$, and Lemma~\ref{lemma:optimal} establishes the global optimality of local optimum for  strictly quasi-concave functions. Then,  Theorem~\ref{th:omega-concave} further shows that the objective function is concave in $\Omega$. These three results lay the theoretical foundation of the proposed algorithm in Section~\ref{sec:algorithm}.

\begin{theorem}
\label{th:quasi-concave}
The objective function $U(\mathbf{p}, \Omega)$ in Problem~\ref{eqn:orig_opt_prb} is a strictly quasi-concave function with respect to the power allocation vector $\mathbf{p}$. Thus, for a given LNA gain $\Omega$, the original optimization problem~\eqref{eqn:orig_opt_prb}  becomes a constrained quasi-concave optimization problem as follows:
\begin{align}
\label{eqn:qc_opt_prb}
\notag & \maximize_{\mathbf{p}}\ U(\mathbf{p}, \Omega)\\
& \subj \ \eqref{limitation1},\,\eqref{limitation2},\,\eqref{limitation3}.
\end{align}
\end{theorem}

\begin{proof}
See Appendix~\ref{apd:1}.
\end{proof}

We note that any strictly concave function is also strictly quasi-concave but the reverse is generally not  true. Moreover, an important property of an optimization problem whose objective function is strictly quasi-concave is that a local optimal solution must be the global optimal one, as formally presented in Lemma~\ref{lemma:optimal}. The proof can be found in  \cite{martos1965direct}.


\begin{lemma}
\label{lemma:optimal}
Suppose function $f$ is strictly quasi-concave. Then the local maximum of $f$ is also the global maximum.
\end{lemma}

%

Note that all the properties we have established so far are for a \emph{fixed} LNA gain $\Omega$. The following theorem establishes the influence of  $\Omega$ on the maximum energy efficiency under all feasible power vectors.

\begin{theorem}
\label{th:omega-concave}
The energy efficiency function $U(\Omega)$ is concave in the LNA gain $\Omega$.
\end{theorem}

\begin{proof}
See Appendix~\ref{apd:2}.
\end{proof}

\section{The B-GAIP Algorithm for Energy Efficiency Optimization}
\label{sec:algorithm}
In this section, we propose the Bisection -- Gradient Assisted Interior Point (B-GAIP) algorithm that solves Problem~\eqref{eqn:orig_opt_prb} under both small and large system dimensions. This algorithm is essentially a two-step implementation of \eqref{eqn:seq} as follows. First, we fix the LNA gain $\Omega$ and design a gradient assisted interior-point (GAIP) algorithm to optimize the power vector, leveraging the strict quasi-concavity property established in Theorem~\ref{th:quasi-concave}. On top of GAIP, we use a bisection search method to find the optimal LNA gain for the maximum energy efficiency, based on its concavity in $\Omega$ as shown in Theorem~\ref{th:omega-concave}. In addition to the detailed description of the proposed algorithms, we  present the proof of convergence and analyze  the algorithm complexity.



\subsection{GAIP: Optimizing Power Allocation under A Fixed LNA Gain}
\label{subsec:alg1}

The heuristic gradient or gradient-based optimization methods are commonly used in energy efficient power allocation problems \cite{Zhao2017Energy,Miao2010Energy,Zhao2013Energy}. This method is well-known and widely-used due to its effectiveness and succinctness. However, under the problem setting of this paper, the optimization objective is a strictly quasi-concave function with convex constraints. Note that \eqref{limitation1} and \eqref{limitation3} are limitations on a single UE transmit power and the total number of these constraints is $2K$, while \eqref{limitation2} is a combination of all the UE transmit powers and the total number of the limitations is $M$. As the system dimension becomes large, so does the number of constraints on the power allocation vector and LNA gain. Therefore, a straightforward adoption of the gradient descent 
algorithm \cite{boyd2004convex} may have very slow convergence, or not converge at all within a reasonable time period of solving the problem. 

To cope with this challenge, especially to make the algorithm efficient and applicable for large system dimensions, we resort to the interior-point method \cite{potra2000interior,boyd2004convex}. The interior-point method is an optimization algorithm which transfers constrained optimization problems into unconstrained ones. The main idea is to construct a penalty function which ``punishes'' the objective function when it approaches or falls out of the boundary of the feasible set. In particular, we chose a logarithmic penalty function in our problem, which is concave. Since we consider a fixed LNA gain $\Omega$ in this subsection, we simply write $U(\mathbf{p},G)$ as $U(\mathbf{p})$ for convenience. The penalty function can then be written as
\begin{align}
\label{eqn:penalty function}
\notag\varphi(\mathbf{p},\xi)&=U(\mathbf{p})+\xi{}B(\mathbf{p})\\
\notag=&U(\mathbf{p})+\xi\sum_{k=1}^K\Big[\ln{}p_k+\ln(P_{\text{max}}-p_k)\\
\notag&+\ln\left(\Gamma_{\text{max}}-\frac{\Omega{}p_{k}}{(\Omega\sigma_N^2+\sigma_{\text{ADC}}^2)\|\mathbf{f}_k\|^2}\right)\Big]\\
&+\xi\sum_{m=1}^M\ln\left[P_{\text{max}}^{\text{ADC}}-\Omega\left(\sum_{k=1}^K|g_{mk}|^2p_k+\sigma_N^2\right)\right].
\end{align}

Note that $B(\mathbf{p})$ represents the penalty for approaching the boundaries, while $\xi$ is the penalty factor which decides the  intensity of penalty. We can see  from \eqref{eqn:penalty function} that when  $\mathbf{p}$ is about to violate the constraints in \eqref{limitation1}, \eqref{limitation2} and \eqref{limitation3}, the latter two terms of \eqref{eqn:penalty function} will reduce the original objective function $U(\mathbf{p})$ by a value that is inversely proportional to the distance between $\mathbf{p}$ and the boundary of the feasible set. Intuitively, as the penalty factor $\xi$ approaches to zero, the penalty function $\varphi(\mathbf{p},\xi)$ is approaching to $U(\mathbf{p})$ as well.

\begin{algorithm}
\label{alg:GAIP}
\SetAlgoNoLine
\caption{Gradient Assisted Interior Point Method}
\KwData{initial penalty factor $\xi^{(0)}$; coefficient $c$; error limit $\epsilon$; maximum loop count $L_\text{max}$;  step size $t_l$}
\KwIn{$\Omega$, channel coefficients}
\KwOut{$\mathbf{p}_\text{opt}$ and $U_{\text{opt}}=U(\mathbf{p}_\text{opt})$}
Randomly choose the initial power vector $\mathbf{p}^{(0)}$ from the feasible set\;
Set initial penalty function value  $\varphi(\mathbf{p}^{(0)},\xi^{(0)})$ using \eqref{eqn:penalty function}\;
Set iteration index $i=0$\;
\SetKwRepeat{doWhile}{do}{while}
\doWhile{$\left|\frac{\varphi(\mathbf{p}^{(i)},\xi^{(i)})-\varphi(\mathbf{p}^{(i-1)},\xi^{(i-1)})}{\varphi(\mathbf{p}^{(i-1)},\xi^{(i-1)})}\right| > \epsilon$}
{$\mathbf{p}_{\text{curr}}=\mathbf{p}^{(i)}$; $\varphi_{\text{opt}}=\varphi(\mathbf{p}^{(i)},\xi^{(i)})$\;
	\For{$l=1$ \KwTo $L_\text{max}$}{
    	Calculate $\mathbf{g}_l=\nabla{}\varphi(\mathbf{p}_{\text{curr}},\xi^{(i)})/|| \nabla{}\varphi(\mathbf{p}_{\text{curr}},\xi^{(i)}) ||$\;
       Update power vector as $\mathbf{p}_{\text{next}}=\mathbf{p}_{\text{curr}}+t_l\mathbf{g}_l$\;        
        \If{$\varphi(\mathbf{p}_\mathrm{next},\xi^{(i)})>\varphi_{\mathrm{opt}}$}
            { 
                Set $\mathbf{p}_\text{curr}=\mathbf{p}_\text{next}$\;
                Set $\varphi_{\text{opt}}=\varphi(\mathbf{p}_\text{next},\xi^{(i)})$\;
            }
    }
    $i$++\; $\xi^{(i)}=\xi^{(i-1)}*c$\;
    $\mathbf{p}^{(i)}=\mathbf{p}_\text{opt}$; $\varphi(\mathbf{p}^{(i)},\xi^{(i)})=\varphi_{\text{opt}}$;}
\end{algorithm}

We then resort to the gradient descent method to find the optimal value of the unconstrained optimization function in \eqref{eqn:penalty function}. In particular, the partial derivative of the penalty function with respect to $p_k$ can be derived as
\begin{align}
\label{eqn:partial derivative}
\notag\frac{\partial{}\varphi(\mathbf{p},\xi)}{\partial{}p_k}&=\frac{\partial{}U(\mathbf{p})}{\partial{}p_k}+\xi\left(\frac{1}{p_k}-\frac{1}{P_{\text{max}}-p_k}-\frac{T_k}{\Gamma_{\text{max}}-T_k{}p_k}\right)\\
&+\xi\sum_{m=1}^M\frac{-\Omega{}|g_{mk}|^2}{P_{\text{max}}^{\text{ADC}}-\Omega\left(\sum_{k=1}^K|g_{mk}|^2p_k+\sigma_N^2\right)},
\end{align}
where we define $T_k=\frac{\Omega{}}{(\Omega\sigma_N^2+\sigma_{\text{ADC}}^2)\|\mathbf{f}_k\|^2}$, and the first term $\frac{\partial{}U(\mathbf{p})}{\partial{}p_k}$ in \eqref{eqn:partial derivative} is given by

\begin{equation}
\frac{\partial{}U(\mathbf{p})}{\partial{}p_k}=\frac{A_d{}T_k}{\ln2(1+A_d\Gamma_k)(P_c+P_{\text{sum}})}-\frac{R_{\text{sum}}}{\eta(P_c+P_{\text{sum}})^2}.
\end{equation}

We further define the gradient metric over power vector $\mathbf{p}$ as $\nabla{}\varphi(\mathbf{p},\xi)=[\frac{\partial{}\varphi}{\partial{}p_1},\ldots,\frac{\partial{}\varphi}{\partial{}p_K}]$. Finally, the proposed GAIP algorithm is compactly presented in Algorithm \ref{alg:GAIP}.

\subsection{B-GAIP: Optimizing Both LNA Gain and Power Allocation}
\label{subsec:alg2}

The GAIP algorithm presented in \ref{subsec:alg1} only optimizes the power values under a fixed LNA gain. A naive approach would be to apply Algorithm~\ref{alg:GAIP} to \emph{all} possible values of $\Omega$, i.e., sweeping all integer values between $\Omega_\text{min}^\text{dB}$ and $\Omega_\text{max}^\text{dB}$, and obtain the optimal energy efficiency. However, this approach may have high complexity if the set of feasible $\Omega$ is large, and it does not utilize the concavity property of the objective function with respect to $\Omega$. 

Alternatively, we propose to solve this optimization problem using a bisection search method, which has lower complexity than linear sweeping, achieves the same optimal value, and leverages the concavity to guarantee optimality (see Theorem \ref{th:omega-concave}). The overall algorithm that solves Problem~\eqref{eqn:orig_opt_prb} is presented in Algorithm \ref{alg:Bisec}.

\begin{algorithm}
\label{alg:Bisec}
\SetAlgoNoLine
\caption{The B-GAIP Algorithm for Solving Problem~\eqref{eqn:orig_opt_prb}}
\KwIn{$\Omega_\text{min}^\text{dB}$, $\Omega_\text{max}^\text{dB}$ and channel coefficients}
\KwOut{optimal power allocation vector $\mathbf{p}_\text{opt}$; optimal LNA gain $\Omega_\text{opt}^\text{dB}$; global maximum energy efficiency $U_\text{max}$}
Set $\Omega_\text{left}^\text{dB}=\Omega_\text{min}^\text{dB}$ and $\Omega_\text{right}^\text{dB}=\Omega_\text{max}^\text{dB}$\;
\While{$\Omega_\mathrm{left}^\mathrm{dB}\neq\Omega_\mathrm{right}^\mathrm{dB}$} 
	{ 
		$\text{LB}=\lfloor(\Omega_\text{left}^\text{dB}+\Omega_\text{right}^\text{dB})/2\rfloor$\;
        		$\text{UB}=\lceil(\Omega_\text{left}^\text{dB}+\Omega_\text{right}^\text{dB})/2\rceil$\;
             
        \If{$\text{LB} == \text{UB}$} 
            { 
                Set $\text{UB}=\text{UB}+1$\; 
            }
            
            $(\mathbf{p}_\text{opt1}, U_{\text{opt1}}) \leftarrow$  \textbf{Algorithm 1} with input $\Omega^{\text{dB}}=\text{LB}$ and channel coefficients\;
            $(\mathbf{p}_\text{opt2}, U_{\text{opt2}}) \leftarrow$  \textbf{Algorithm 1} with input $\Omega^{\text{dB}}=\text{UB}$ and channel coefficients\;

        \eIf{$U_{\mathrm{opt1}}>U_{\mathrm{opt2}}$}
            { 
                Set $\Omega_\text{right}^\text{dB}=\text{LB}$\;
            }
            {
                Set $\Omega_\text{left}^\text{dB}=\text{UB}$\;
            }
        
        \If{$\Omega_\text{left}^\text{dB}==\Omega_\text{right}^\text{dB}$}
        	{
            	Set $\Omega_\text{opt}^\text{dB}=\Omega_\text{left}^\text{dB}$\;
            	Set $U_\text{max}=\max\{U_{\text{opt1}},U_{\text{opt2}}\}$\;
                	Choose $\mathbf{p}_\text{opt}$ according to $U_\text{max}$\;
        	}
    }
\end{algorithm}


\subsection{Analysis of the B-GAIP Algorithm}
\label{subsec:Evaluations}
In this section, we will theoretically evaluate the convergence and complexity of the proposed algorithm. In order to do so, we first reorganize the B-GAIP algorithm into an equivalent formulation in Sec.~\ref{sec: ana1}. This new interpretation, although does not bear algorithmic novelty, simplifies the convergence and complexity study in Sec.~\ref{sec: ana2} and \ref{sec: ana3}.


\subsubsection{An Equivalent Interpretation of B-GAIP}
\label{sec: ana1}
There are three obstacles to solve the original optimization problem \eqref{eqn:qc_opt_prb}: (1) under a fixed LNA gain, the objective function is  strictly quasi-convex, which is difficult to handle compared to the standard convex function; (2) for massive MIMO, the system dimension is large and thus the algorithm complexity is of great importance to its practical utility; and (3) the engineering constraints must be satisfied. Therefore, the methods in existing literature \cite{Zhao2017Energy,Miao2010Energy,chen2012energy} are not applicable in our problem setting. 

In order to cope with these obstacles, we have proposed the B-GAIP algorithm in this paper. For the convenience of the theoretical evaluations, we re-interpret Algorithm~\ref{alg:Bisec} and provide some important comments in the following.

In particular, the procedure of B-GAIP can be reorganized as follows.

\begin{itemize}
\item Step-1: Choose two LNA gain values, $\text{LB}$ and $\text{UB}$, as described in Algorithm \ref{alg:Bisec}.
\item Step-2: Under these two gains, use Algorithm \ref{alg:GAIP} to find the maximal energy efficiencies, $U_\text{opt1}$ and $U_\text{opt2}$, respectively.
\begin{itemize}
\item Step-2.1: In Algorithm \ref{alg:GAIP}, choose the penalty factor and convert the problem into an unconstrained quasi-concave optimization problem.
\item Step-2.2: Solve the converted problem using the gradient descent method.
\item Step-2.3: Go back to Step-2.1 with a smaller penalty factor if the accuracy requirement is not met; otherwise return the maximal energy efficiency and the optimal power vector.
\end{itemize}
\item Step-3: Compare the two outputs of Step-2 and then decide if the algorithm converges. Go back to Step-1 if not converge; otherwise return the optimal power allocation vector and LNA gain.
\end{itemize}

From the procedure above, we can conclude that the two-step algorithm has three main layers as follows.
\begin{itemize}
\item The outer layer: Optimize the LNA gain via the bisection search method;
\item The middle layer: Under a given LNA gain, transfer the constrained problem into an unconstrained optimization problem via the interior-point method;
\item The inner layer: Find the optimal value of the unconstrained problem via the gradient descent method.
\end{itemize}
We will use this interpretation in the following convergence and complexity analysis.

\subsubsection{The Convergence Analysis}
\label{sec: ana2}
It is difficult to directly analyze the convergence of B-GAIP due to its inherent complexity. We thus leverage the equivalent interpretation and separately study the convergence of each layer, and the overall convergence of B-GAIP can be proved by a combination argument. In the following,  the convergences of the inner, middle and outer layers will be presented respectively.

The gradient descent method is used in the inner layer, with $L_{max}$ loops be performed while at each loop $l$, the argument $\mathbf{p}$ move towards the gradient direction which is also the ascent direction of function $\varphi(\mathbf{p})$ with a step length that diminishes in $l$. Note that only the direction of the gradient is used due to normalization in our algorithm. As a result, the step length that is moved at each loop, i.e., $||\mathbf{p}_\text{next}-\mathbf{p}_\text{curr}||$, is exactly the step size $t_l$, where we set $t_l=0.01/l$. When $L_{max} \to \infty$, we have $t_l \to 0$ while $\sum_{l=1}^{\infty}t_l=\infty$, which suggests that the argument can move an unbounded distance towards the optimal value if sufficient iterations are allowed, thus guaranteeing the convergence. Recall that Lemma~\ref{lemma:optimal} proves that the local maximum of our objective function is also the global maximum, and thus we conclude that the inner layer not only converges, but also converges to the global optimum.

For the middle layer, we use the interior-point method to transfer the original constrained optimization problem into an unconstrained one, and then invoke the gradient descent method (the inner layer) to solve it. Since we have already proved the convergence of the inner layer, it is now sufficient to establish the convergence of the middle layer if we can prove the equivalence of the original constrained problem and the converted unconstrained problem. We start the argument by noting from \eqref{eqn:penalty function} that the difference between the original and the converted functions is the boundary penalty $B(\mathbf{p})$, which is controlled by the penalty factor $\xi$. At the end of each loop, the penalty factor $\xi$ is multiplied by a decreasing coefficient $c$, which eventually results in $\xi \to 0$, and the converted function $\varphi(\mathbf{p},\xi)$ then converges to the original objective function $U(\mathbf{p})$. However, as the penalty factor $\xi \to 0$, the difficulty for the inner layer to obtain the optimal value increases as it may approach the boundary of the feasible set. This issue can be resolved by using the output of the previous iteration as the starting point for the new iteration \cite{boyd2004convex}, as is done in Algorithm \ref{alg:GAIP}. As a result, we can achieve the global optimal value $\mathbf{p}_\text{opt}$ for maximizing $U(\mathbf{p})$ by combining the inner and middle layers. 

Finally, in the outer layer, the bisection search method is used to optimize the LNA gain. We have already proved that the optimal power allocation vector under \emph{any} fixed LNA gain is achieved. Because of the concavity property (Theorem~\ref{th:omega-concave}) and the optimality of the bisection search method \cite{boyd2004convex}, we conclude that the algorithm must converge to the optimum. In fact, at most $\log_2\left(\Omega_\text{max}^\text{dB}-\Omega_\text{min}^\text{dB}\right)$ iterations will be performed before we reach the optimal value.


\subsubsection{The Complexity Analysis}
\label{sec: ana3}
In order to quantitatively study the complexity of Algorithm \ref{alg:GAIP} and \ref{alg:Bisec}, we individually analyze the complexity of each layer like we did in the convergence analysis. For the inner layer, we perform $L_{max}$ iterations and within each iteration, the partial derivative of each UE is calculated separately, resulting in a complexity scaling $\mathcal{O}(KL_{max})$. For the middle layer, the number of iterations will change according to the required accuracy, and therefore, it is a function of the error limit $\epsilon$. We denote the number of iteration times as $T_\epsilon$ and the complexity scaling of this layer should be $\mathcal{O}(T_\epsilon)$. Finally, for the outer layer, the complexity scaling of the bisection search is $\mathcal{O}\left(\log_2\left(\Omega_\text{max}^\text{dB}-\Omega_\text{min}^\text{dB}\right)\right)$. Putting all three layers together, the overall complexity of B-GAIP is of the order: 
\begin{equation}
\label{eqn:complexity}
\mathcal{O}\left(KL_{max}T_\epsilon\log_2\left(\Omega_\text{max}^\text{dB}-\Omega_\text{min}^\text{dB}\right)\right).
\end{equation}

Qualitatively, as discussed before, in the scenario with large number of BS antennas and UEs, it becomes time-consuming to determine whether the boundary limitations are violated. Fortunately, this difficulty is circumvented in Algorithm \ref{alg:GAIP} as it converts the engineering constraints into penalty items and therefore transfers a constrained optimization problem to an unconstrained one, greatly reducing the complexity, especially when the system dimension is large. In the meanwhile, Algorithm \ref{alg:Bisec} utilizes a bisection approach which reduces the search time exponentially compared with the intuitive linear search method. Accordingly, Algorithm \ref{alg:GAIP} and \ref{alg:Bisec} are more efficient than heuristic solutions and applicable for large scale systems.

\section{Simulation Results}
\label{sec:simulation}
We resort to system-level simulations of an uplink MIMO system to numerically evaluate the proposed B-GAIP algorithm. Important simulation parameters can be found in Table \ref{table:parameters}. In particular, for a given MIMO configuration, many realizations of the small-scale fading vector $\mathbf{h}_k$, the $K$ UE positions, and the $M$ BS antenna positions (in the case of distributed MIMO) are randomly generated. 

\begin{table}[!htb]
\begin{center}
\caption{Simulation Parameters}
\label{table:parameters}
\begin{tabular}{| c | c |}
\hline
Cell radius $R_0$ & 100 $\sim$ 1000 m \\
\hline
Background noise $\sigma_N^2$ & -104 dBm \\
\hline
\multirow{2}{*}{Shadow fading $V$} & Log-normal with \\ &standard deviation of 8 dB \\
\hline
ADC noise $\sigma_{\text{ADC}}^2$ & -60 dBm \\
\hline
Minimal LNA gain $\Omega_\text{min}^\text{dB}$ & 1 dB \\
\hline
Maximal LNA gain $\Omega_\text{max}^\text{dB}$ & 70 dB \\
\hline
Diversity gain $A_d$ & 1 \\
\hline
Circuit power $P_c$ & 0.1 W \\
\hline
Power amplifier efficiency $\eta$ & 50\% \\
\hline
Maximal transmit power $P_\text{max}$ & 20 dBm \\
\hline
Maximal ADC input power $P_\text{max}^\text{ADC}$ & -20 dBm \\
\hline
Maximal SNR $\Gamma_\text{max}$ & 35 dB \\
\hline
Step size $t_l$ & $0.01/l$ \\
\hline
\end{tabular}
\end{center}
\end{table}

In the following, we first compare the proposed algorithm with the heuristic brute force search solution, to verify the feasibility and accuracy of the B-GAIP algorithm in Sec. \ref{subsec:algorithm}.
Aim at analyzing the energy efficiency performance in different system settings, in Sec. \ref{subsec:dimension_layout}, we compare the average maximum energy efficiency with centralized and distributed MIMO layouts in both small and large system dimensions. In particular, we evaluate the energy efficiency with different UE number $K$, BS antenna number $M$, and the ratio $v=M/K$ with different system layouts.
After that, we adjust the cell radius from 100 meters to 1000 meters with a 100m step size to evaluate the effect of cell radius on energy efficiency in Sec. \ref{subsec:radius}. Note that the comparisons are made with both centralized and distributed MIMO layouts.
In addition, we compare  B-GAIP  with heuristic algorithms to evaluate the performance gain of the proposed algorithm in Sec. \ref{subsec:heuristic}.
Finally, comparison between shared and separate LNA is carried out in Sec. \ref{subsec:shar_sepa} to illustrate the trade-off between energy efficiency  and  system costs. These numerical simulation result may offer guidance to the design and operation of cost-aware massive MIMO systems.

\subsection{Comparison between B-GAIP and Brute Force Search}
\label{subsec:algorithm}
It is crucial to verify whether the proposed B-GAIP algorithm will converge to the global optimal solution of the original optimization problem. In addition to the theoretical analysis in Section \ref{subsec:Evaluations},  we now compare our algorithm with the naive brute force solution that tries every possible parameter combination to find out the optimal solution. In particular, we choose a small system dimension with 2 UEs and 4 BS antennas in a distributed layout due to the high complexity of brute force search. We try all possible transmit power values and the LNA gain in dB domain with 0.1dB and 1dB step-size, respectively. We change the cell radius from 100m to 1000m and at each radius, we perform 2000 realizations of the channel parameters including UE and BS antenna positions, fast fading and shadow fading. In each realization, we run the proposed B-GAIP algorithm and the brute force solution separately, under the same computational environment, to find out the maximum energy efficiency. We also record the run time in each realization.

\begin{figure}
\centering
\includegraphics[width=0.48\textwidth]{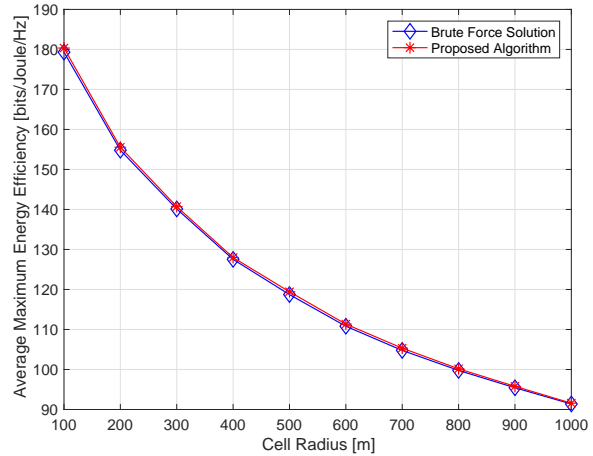}
\caption{Energy efficiency comparison between the B-GAIP algorithm and the brute force  solution.}
\label{fig:compare}
\end{figure}

Fig. \ref{fig:compare} illustrates the average maximum energy efficiency obtained by the B-GAIP algorithm and the brute force search under different cell radii. We can clearly see from the figure that our B-GAIP algorithm has the same energy efficiency as the brute force solution and therefore converges to the optimal value\footnote{Note that while looking the same in the figure, the average maximum energy efficiency in B-GAIP actually slightly exceeds that of the brute force search, which is the result of the 0.1dB step-size since it may skip over the true optimal value.}. 

Moreover, we compare the complexity between the proposed algorithm and other heuristic brute force search solutions, by comparing their running time.  We compare B-GAIP with three algorithms: (1) naive brute force solution, which simultaneously searches  power vector $\mathbf{p}$ and LNA gain $\Omega$ for the optimal values; (2) hybrid solution type 1 (Hybrid1), which uses brute force search for $\Omega$, while still uses the interior-point method to find the optimal $\mathbf{p}$; (3) hybrid solution type 2 (Hybrid2), which uses brute force solution to search $\mathbf{p}$, while uses bisection search to find the optimal $\Omega$. The results are shown in Table \ref{table:time}\footnote{Note that the computational environment will affect the time consumptions. We carry out the comparison under the same computational environment, and the results are also averaged over large amount of realizations.}.

\begin{table*}[!htb]
\begin{center}
\caption{Complexity Comparison of Different Algorithms}
\label{table:time}
\begin{tabular}{|c|c|c|c|c|}
\hline
\diagbox{UE numbers}{Time [s]}{Algorithms} & B-GAIP & \tabincell{c}{Naive\\Brute Force} & \tabincell{c}{Hybrid1:\\Brute Force: $\Omega$\\ \& Interior-Point: $\bf{p}$} & \tabincell{c}{Hybrid2:\\Brute Force: $\bf{p}$\\ \& Bisection: $\Omega$}\\
\hline
1 & 0.0345 & 0.0195 & 0.2237 & 0.0026\\
\hline
2 & 0.1087 & 4.2681 & 0.6989 & 0.4854\\
\hline
3 & 0.1210 & 591.5498 & 0.8010 & 70.1647\\
\hline
4 & 0.1307 & 70590.9569 & 0.8852 & 8696.3573\\
\hline
5 & 0.1359 & $\backslash$ & 0.9625 & $\backslash$\\
\hline
6 & 0.1390 & $\backslash$ & 1.0561 & $\backslash$\\
\hline
7 & 0.1432 & $\backslash$ & 1.1237 & $\backslash$\\
\hline
8 & 0.1459 & $\backslash$ & 1.1867 & $\backslash$\\
\hline
9 & 0.1490 & $\backslash$ & 1.2415 & $\backslash$\\
\hline
10 & 0.1527 & $\backslash$ & 1.3107 & $\backslash$\\
\hline
\end{tabular}
\end{center}
\end{table*}

We can see from the table that both B-GAIP and Hybrid1 algorithms outperform the other two significantly. It is a direct consequence of the fact that brute force search of $\bf{p}$ has an exponential complexity in the number of UEs. We can see that the superiority of B-GAIP over Hybrid1 is mainly due to the use of bisection search method, which reduces the time complexity of optimizing $\Omega$ from a linear order to a logarithmic order.

In addition to the general comparison of complexity scaling, we further look into the effect of number of UEs in the system.  Note that the time consumption of B-GAIP and Hybrid1 when there is only one UE is longer that the others, which is a consequence of the fact that the initialization time of the interior-point method is rather significant. Furthermore, the complexities of B-GAIP and Hybrid1 algorithms are approximately increasing linearly with the number of UEs, which agrees with the analysis in Section \ref{sec:algorithm}.

\subsection{Energy Efficiency with Different Dimensions and Layouts}
\label{subsec:dimension_layout}
We now turn our attention to evaluating the energy efficiency performance under different system dimensions and layouts. There are several parameters to adjust, including the number of UEs $K$, the number of BS antennas $M$, the ratio $v=M/K$ and the different network topologies, i.e., centralized or distributed MIMO layouts. We will evaluate the influences of these factors respectively.

\begin{figure*}[!htb]
    \centerline{
        \subfigure[ Large system dimension ]{ 
        \includegraphics[width=0.49\textwidth]{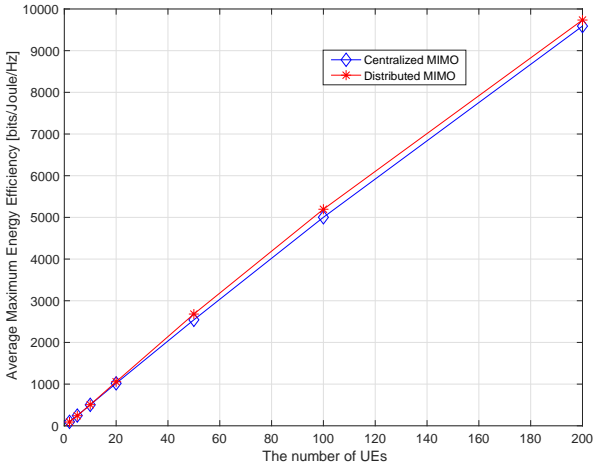}
        \label{fig:keepv_large}}
        \hfil
        \subfigure[ Small system dimension ]{         
        \includegraphics[width=0.49\textwidth]{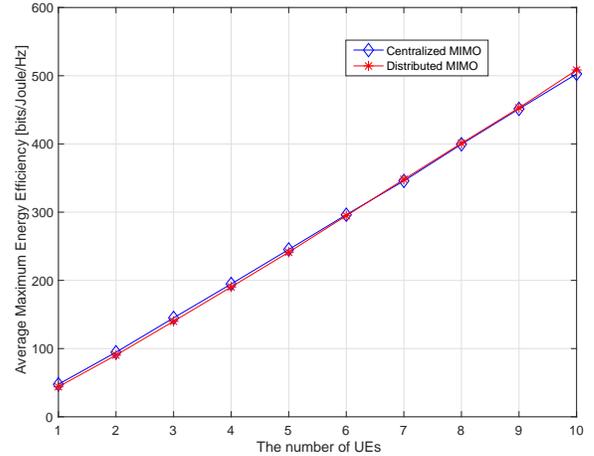}
        \label{fig:keepv_small}}}
    \caption{The effect of system dimension on  energy efficiency with a fixed ratio $v=M/K=2$.}
    \label{fig:keepv}
\end{figure*}

Fig. \ref{fig:keepv_large} shows the effect of system dimensions on the energy efficiency, with both centralized and distributed MIMO layouts. In this simulation, we let $K$ and $M$ grow simultaneously with a fixed ratio $v=M/K=2$. We can observe that the average maximum energy efficiency increases approximately \emph{linearly} with the system dimension. Note that even though the energy efficiency in a distributed MIMO layout outperforms the centralized one when $K$ and $M$ are large, it is actually the opposite when the system dimension is small, as depicted in Fig. \ref{fig:keepv_small} which shows the energy efficiency performance when the number of UEs grows from 1 to 10, while the BS antennas are twice as many as the number of UEs. Intuitively, the phenomenon in Fig. \ref{fig:keepv} can be explained as follows. When the numbers of transmitters and receivers are both small, the average access distance in a centralized layout is  shorter than that the distributed layout. This is because the BS antennas are all co-located in the center of the coverage area in a centralized layout, while in the distributed layout, BS antennas are randomly distributed throughout the cell, which may occasionally result in larger access distance when the number of antennas is very small. On the contrary, when the number of antennas grows, the average access distance in a distributed layout will become shorter than that in centralized layout, which means less energy consumed while delivering the same amount of information. Therefore, the energy efficiency in a distributed layout will eventually outperform the centralized layout as the system dimension grows.

\begin{figure*}[!htb]
    \centerline{
        \subfigure[ $M=400$, $K=2 \sim 200$ ]{ 
        \includegraphics[width=0.49\textwidth]{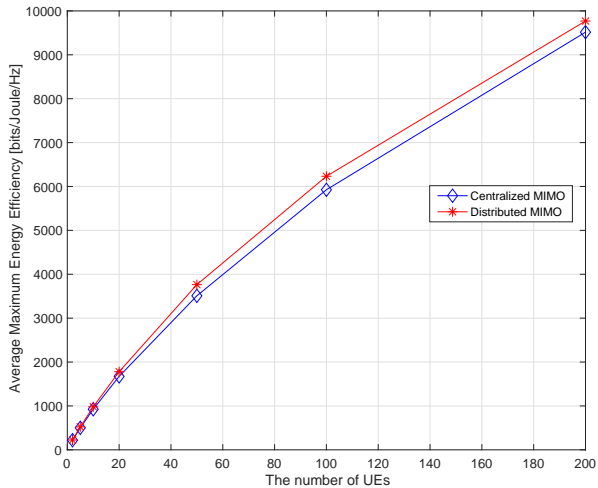}
        \label{fig:changeK}}
        \hfil
        \subfigure[ $K=2$, $M=5 \sim 500$ ]{         
        \includegraphics[width=0.49\textwidth]{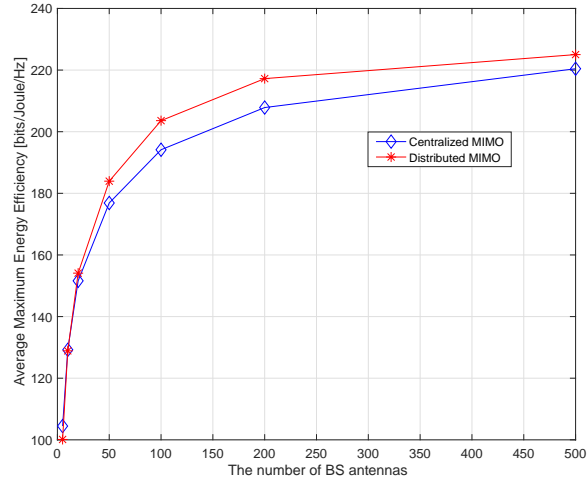}
        \label{fig:changeM}}}
    \caption{The effect of $K$ and $M$ on energy efficiency.}
    \label{fig:changeKM}
\end{figure*}

Fig. \ref{fig:changeKM} further illustrates the effects of $K$ and $M$ on energy efficiency, respectively. We see from Fig. \ref{fig:changeK}  that the number of UEs  has a significant impact on the energy efficiency, while Fig. \ref{fig:changeM} implies that the number of BS antennas has much less effect, especially when $M$ reaches a relatively large value. As a result, the dominating factor that determines the average maximum energy efficiency is $K$. 

\subsection{The Effect of Cell Radius}
\label{subsec:radius}

In order to study the effect of cell radius on energy efficiency, we plot Fig. \ref{fig:radius} with different antenna layouts. It is clear that the average maximum energy efficiency decreases as the cell radius increases, which is also observed in some existing papers such as \cite{Miao2010Energy,Zhao2017Energy}. The curve is convex, which suggests that the decrease of the energy efficiency is large when the cell radius is small and the trend will slow down as the cell radius grows. As a side note, we set $K=2$ and $M=4$ in this simulation, which is a relatively small system dimension, and the energy efficiency performance in the centralized layout will outperform the distributed layout. Since a large cell radius means that more transmit power will be needed to convey the same amount of information, it is clear that the energy efficiency will decrease with the growth of the cell radius.

\begin{figure}
\centering
\includegraphics[width=0.48\textwidth]{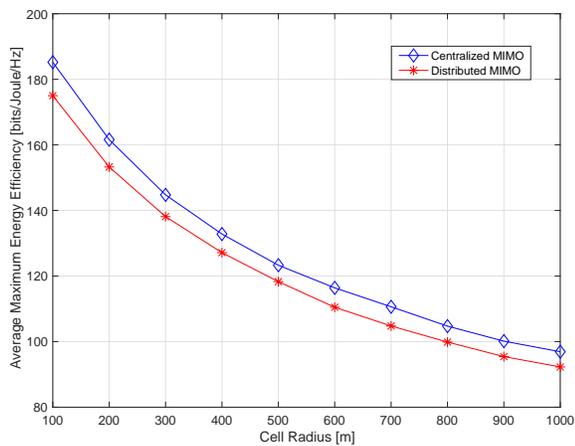}
\caption{The effect of cell radius on energy efficiency, with $K=2$ and $M=4$.}
\label{fig:radius}
\end{figure}

\subsection{Comparison between Proposed Algorithm and Heuristic Algorithms}
\label{subsec:heuristic}
In addition to analyzing the energy efficiency performance of our proposed algorithm, we also compare B-GAIP with two commonly adopted heuristic algorithms \cite{Zhao2017Energy} to evaluate their performance differences in the distributed MIMO setting. Intuitively, higher transmit power and LNA gain shall result in higher SNR, and therefore higher energy efficiency performance. Correspondingly, we consider the following two heuristic algorithms: (1) use the maximal LNA gain combined with the corresponding optimal transmit power vector; and (2) use the maximal transmit power vector combined with the corresponding optimal LNA gain. Note that the engineering constraints in \eqref{limitation1}, \eqref{limitation2}, and \eqref{limitation3} are still enforced when using these heuristic algorithms.

\begin{figure*}[!htb]
    \centerline{
        \subfigure[ $K=2$, $M=4$ ]{ 
        \includegraphics[width=0.49\textwidth]{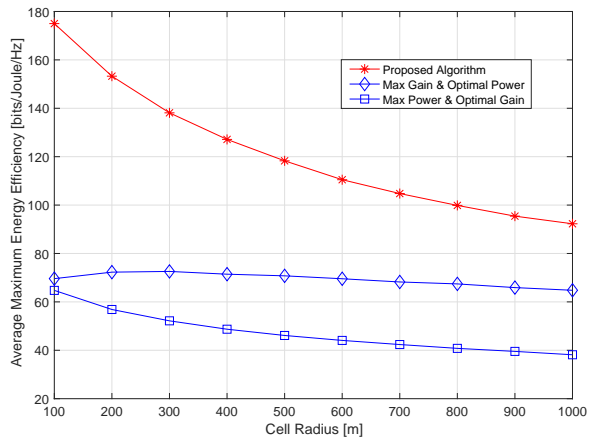}
        \label{fig:Heuristic_Small}}
        \hfil
        \subfigure[ $K=10$, $M=20$ ]{         
        \includegraphics[width=0.49\textwidth]{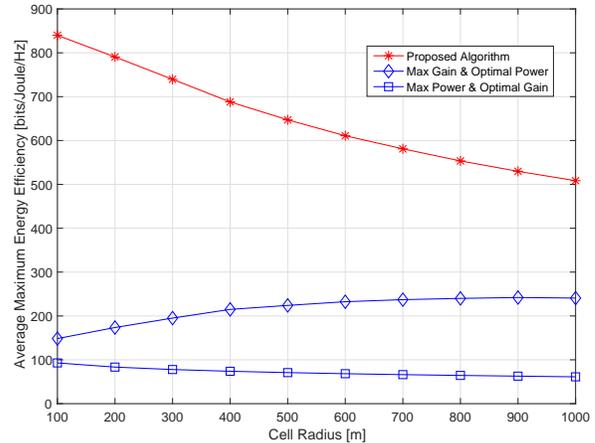}
        \label{fig:Heuristic_Moderate}}}
    \caption{Comparison between B-GAIP and the heuristic algorithms, in both small and moderate system dimensions.}
    \label{fig:Heuristic}
\end{figure*}

Fig. \ref{fig:Heuristic} shows the energy efficiency performance of three algorithms under both small and moderate system dimensions. The performance gap is quite significant, suggesting that invoking B-GAIP to solve the original optimization problem is necessary. In addition, two other important observations can be drawn from Fig. \ref{fig:Heuristic}.

Fig. \ref{fig:Heuristic_Small} depicts a situation where the system dimension is small. It is clear that our proposed algorithm outperforms the heuristic ones with a large margin. At the maximum point, the average maximum energy efficiency value achieved by B-GAIP is 151.4\% higher than that of the heuristic algorithm; while at the minimum point, the advantage becomes 42.4\%, which is still a considerable gain. Meanwhile, we expect an even higher performance gain of B-GAIP when the system dimension becomes larger, since the bounded subspace of the tunable parameters will become larger, which means that there will be more options available for the algorithm to optimize with. To validate this conjecture, we run the simulation and report the corresponding result in Fig. \ref{fig:Heuristic_Moderate}, where a moderate system dimension is used as an example to compare the performance change. We can see from the figure that even at the minimum point, the performance of B-GAIP is still more than twice as that of the heuristic algorithm.

Different from the conclusion we draw in Sec. \ref{subsec:radius} that the energy efficiency  will decrease with the growth of cell radius, by using the heuristic algorithm \textit{with maximum LNA gain}, Fig. \ref{fig:Heuristic_Small} shows that the energy efficiency actually first increases and then decreases, and Fig. \ref{fig:Heuristic_Moderate} indicates that the energy efficiency will always increase with the cell radius. While at the first glance, these results seem counterintuitive and may even contradict the analysis in subsection \ref{subsec:radius}, a deeper investigation can explain this phenomenon as follows. 
On the one hand, the cell radius does affect the energy efficiency. If we keep all other parameters, such as LNA gain and transmit power vector, unchanged and only increase the cell radius, the energy efficiency will decrease. However, the constraints in the optimization problem may affect the energy efficiency in a different way. Note that  the maximal LNA gain is used here, which may cause the ADC to saturate rather frequently, and thus reach the maximum SNR limitation . As a result, the main factor that keeps the performance from further growing is the SNR limitation rather than the cell radius in most cases. 

This explanation is numerically validated in Fig. \ref{fig:Heuristic_Moderate} where the average access distance is relatively small and so is the pass loss. In this case, a large LNA gain does not help. The LNA saturation can be mitigated when the cell radius becomes larger, and the performance will increase accordingly. On the other hand, for a small system dimension as Fig. \ref{fig:Heuristic_Small} shows, the antennas are not very crowded even when the cell radius is relatively small. Hence, as the cell radius increases, the  performance first reaches the maximum, and any further-increasing radius will then degrade the energy efficiency performance. 

\subsection{Comparison between Shared LNA and Separate LNA}
\label{subsec:shar_sepa}

\begin{figure*}
    \centerline{
        \subfigure[ Centralized MIMO ]{ 
        \includegraphics[width=0.49\textwidth]{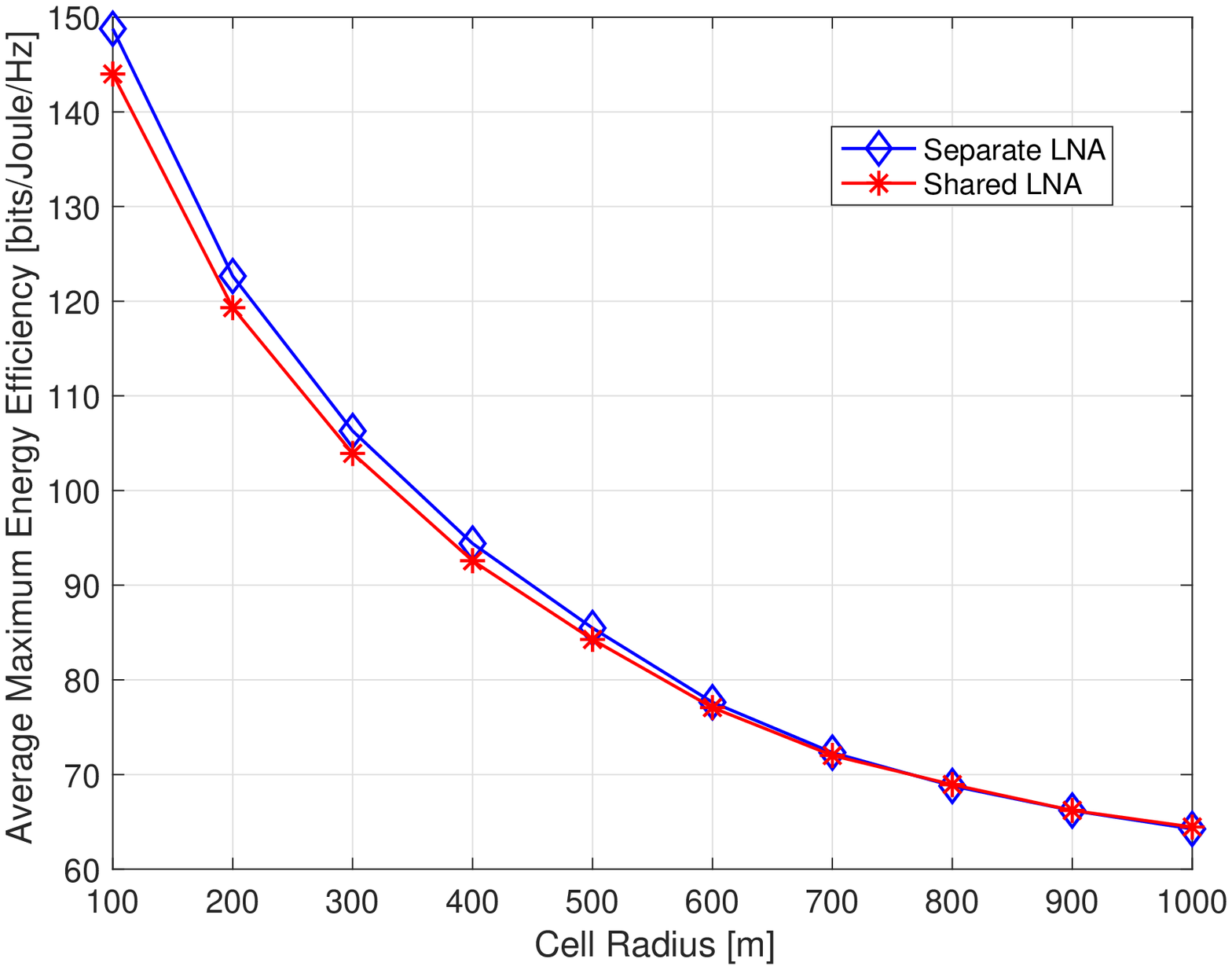}
        \label{fig:SepvsSha_CA}}
        \hfil
        \subfigure[ Distributed MIMO ]{         
        \includegraphics[width=0.49\textwidth]{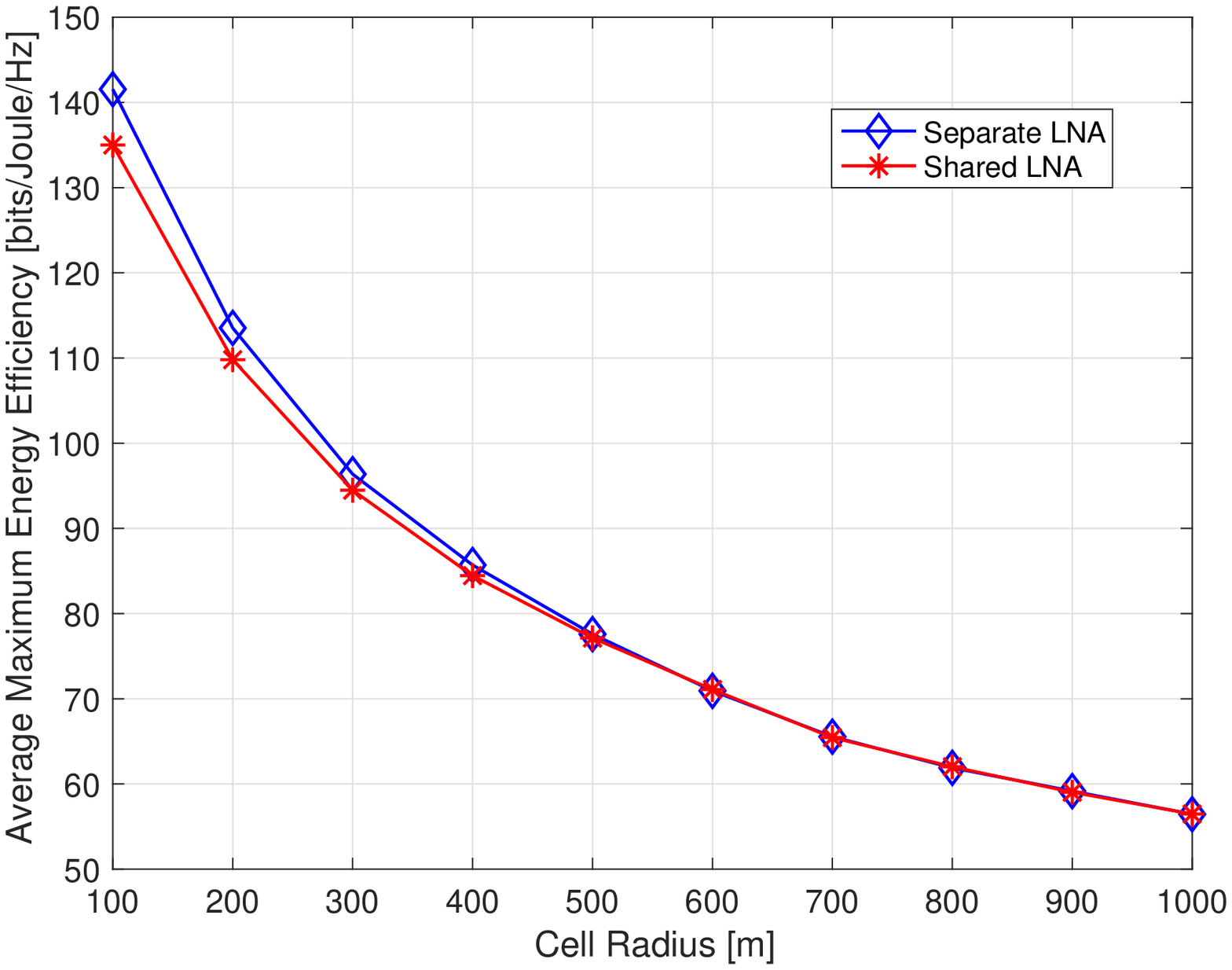}
        \label{fig:SepvsSha_DA}}}
    \caption{Comparison between shared LNA and separate LNA structures under both centralized and distributed MIMO layouts.}
    \label{fig:SepvsSha}
\end{figure*}

Our system model uses one shared LNA to amplify the BS received signals in order to save both implementation cost and power consumption. One natural question is how much performance sacrifice we are incurring compared with using a separate LNA for each RF chain.  In this subsection, we aim to address this question via system simulations.

We first derive energy efficiency under the same system model but with separate LNA at each RF chain. For clarity, we use subscript ``sep'' to denote the signals using the separate LNA structure throughout this subsection. The received signal vector is the same as \eqref{eqn:received signal}, while the amplified signal can be written as
\begin{equation}
\tilde{\mathbf{y}}_{\text{sep}}=\mathbf{\Omega}\mathbf{y},
\end{equation}
where $\mathbf{\Omega}=$ diag($\sqrt{\Omega_1}$ , \ldots , $\sqrt{\Omega_M}$) denotes the LNA gain values  of the $M$ separate LNAs.

After passing $\tilde{\mathbf{y}}_{\text{sep}}$ through separate ADCs for quantization, we also adopt a ZF receiver to process the quantized signal by multiplying it with the pseudo-inverse of the equivalent channel matrix $\hat{\mathbf{G}}=\mathbf{\Omega}\mathbf{G}$, where we have $\hat{\mathbf{F}}\hat{\mathbf{G}}=\mathbf{I}_K$. Now $\mathbf{r}_{\text{sep}}$ is given by
\begin{equation}
\mathbf{r}_{\text{sep}}=\mathbf{Px}+\hat{\mathbf{F}}\mathbf{\Omega}\mathbf{z}+\hat{\mathbf{F}}\mathbf{n}_q.
\end{equation}
Like \eqref{eqn:component}, we also take the $k$th component of the vector $\mathbf{r}_{\text{sep}}$ as an example:
\begin{equation}
r_{\text{sep},k}=\sqrt{p_k}x_k+\hat{\mathbf{f}}_k{}\mathbf{\Omega}\mathbf{z}+\hat{\mathbf{f}}_k\mathbf{n}_q,
\end{equation}
where $\hat{\mathbf{f}}_k$ denotes the $k$th row of matrix $\hat{\mathbf{F}}_k$. As a result, the SNR of the $k$th UE can be calculated as
\begin{equation}
\Gamma_{\text{sep},k}=\frac{p_{k}}{\sigma_N^2\hat{\mathbf{f}}_k\mathbf{\Omega}^2\hat{\mathbf{f}}_k^H+\sigma_{\text{ADC}}^2\|\hat{\mathbf{f}}_k\|^2}.
\end{equation}
Finally, the spectral efficiency and the energy efficiency remain the same as \eqref{eqn:se-def} and \eqref{eqn:ee-def}, respectively. Note that the constraints in \eqref{limitation2} and \eqref{limitation3} will change to the following, while \eqref{limitation1} still holds.

\begin{itemize}
\item ADC saturation limitation:
\begin{equation}
\Omega_m(\mathbf{g}_m\mathbf{P}^2\mathbf{g}_m^H+\sigma_N^2)\leqslant{}P_{\text{max}}^{\text{ADC}},\ m\in\mathcal{M}.
\end{equation}

\item Maximum SNR limitation:
\begin{equation}
\frac{p_{k}}{\sigma_N^2\hat{\mathbf{f}}_k\mathbf{\Omega}^2\hat{\mathbf{f}}_k^H+\sigma_{\text{ADC}}^2\|\hat{\mathbf{f}}_k\|^2}\leqslant{}\Gamma_{\text{max}},\ k\in\mathcal{K}.
\end{equation}
\end{itemize}

Note that the objective function and the constraints under the separate LNA model have  similar properties as Theorem \ref{th:quasi-concave}, i.e., $U_\text{sep}(\mathbf{p},\mathbf{\Omega})$ is a strictly quasi-concave function with convex constraints under a fixed LNA gain matrix $\mathbf{\Omega}$. Therefore, we will use a similar method as Algorithm \ref{alg:GAIP} to find the optimal power allocation vector with given $\mathbf{\Omega}$. However, since the function $U_\text{sep}(\mathbf{\Omega})$ under the separate LNA structure does not have the concavity property proved in Theorem \ref{th:omega-concave}, we resort to the brute force search solution to find the optimal LNA gain values, which has a higher complexity. For the ease of numerical simulations, we only compare the performance under a small system dimension where the numbers of UEs and the BS antennas are set to 2, i.e., $M=K=2$.

Fig. \ref{fig:SepvsSha} reports the comparison of energy efficiency with separate and shared LNA structure in both centralized and distributed MIMO layouts, respectively. We conclude from the figure that while the separate LNA structure achieves a better performance, using a shared LNA structure can very closely approach the performance of the separate LNA. Taking a deeper look at the statistics, we have that the maximum performance loss in a centralized layout is 3.21\%, while this loss becomes 4.62\% in a distributed layout\footnote{It is worth noting that we use the same circuit power, i.e., $P_c=0.1$ W, in both LNA structures, while in reality the separate LNA structure should have more power consumption than the shared LNA structure, which may further degrade its energy efficiency.}. As a result, using a shared LNA can significantly reduce the hardware cost and power consumption, while sacrificing very little  energy efficiency. This result sheds important light on the design of RF front-end power amplifiers in practical MIMO systems.

\section{Conclusions}
\label{sec:conclusions}
Energy efficiency is of great importance in modern wireless communications, especially when massive MIMO gains its popularity in 5G systems. The high power consumption of  RF front-end components including LNA and ADC significantly affects the energy efficiency performance in MIMO systems. In this paper, we have proposed a \emph{shared LNA} structure and showed that combined with low-resolution ADCs, this architecture saves both hardware costs and reduces power consumption, while achieving near-optimal performance. In particular, we formulated the energy efficiency maximization problem under real-world engineering constraints, and revealed several important properties of this problem. We then proposed the Bisection -- Gradient Assisted Interior Point (B-GAIP) algorithm that solves the optimization problem precisely and efficiently. The convergence and complexity analysis of B-GAIP have been studied, and comprehensive simulations have been performed to validate the effectiveness of the proposed algorithm. 


Although massive MIMO under realistic hardware constraints have attracted much attention lately, the existing literature, including this paper, still leaves many problems unsolved. For example, the energy efficiency of a massive MIMO system with imperfect CSI under a similar setting is of great importance in practice but remains unsolved. How to incorporate other practical constraints and system design objectives into this problem is another interesting research topic, which may be worth investigation in the future.

\appendices
\section{Proof of Theorem \ref{th:quasi-concave}}
\label{apd:1}
Since our proof is based on a given $\Omega$, we use $U(\mathbf{p})$ instead of $U(\mathbf{p},\Omega)$ to simplify the notation. Furthermore, we define $U(\mathbf{p})$'s super-level set as
\begin{equation}
\label{S_alpha}
S_\alpha=\{\mathbf{p}\succeq\mathbf{0}|U(\mathbf{p})\geqslant\alpha\}.
\end{equation}
By using the equivalent definition of quasi-concavity in \cite{boyd2004convex} that a function $f$ is called strictly quasi-concave if its domain $\mathbb{D}$ and all its super-level sets $S_\alpha=\{\mathbf{x}\in\mathbb{D}|f(\mathbf{x})\geqslant\alpha\}$ are convex, $U(\mathbf{R})$ is strictly quasi-concave if $S_\alpha$ in \eqref{S_alpha} is strictly convex for any real number $\alpha$.

Now it is sufficient to show the quasi-concavity of $U(\mathbf{p})$ if we can prove $S_\alpha$ is convex for positive, negative and zero values of $\alpha$, respectively. Since we have the property that $\mathbf{R}\succeq\mathbf{0}$ and $\mathbf{p}\succeq\mathbf{0}$, $U(\mathbf{p})$ is therefore nonnegative over all possible power vector $\mathbf{p}$. As a result, no points exist on the contour $U(\mathbf{p})=\alpha$ for $\alpha<0$, and only one point $\mathbf{0}$ is on the contour $U(\mathbf{p})=\alpha$ for $\alpha=0$. Hence, $S_\alpha$ is convex when $\alpha\leqslant{}0$. 

To prove the case $\alpha>0$, we first investigate the property of $R_{\text{sum}}\triangleq\sum_{k=1}^{K}R_k$. The proof of the following lemma is straightforward and is omitted.

\begin{lemma}
Recall the definition of $R_{\text{sum}}(\mathbf{p})$:
\begin{equation}
R_{\text{sum}}(\mathbf{p})=\sum_{k=1}^{K}R_k=\sum_{k=1}^{K}\log_2(1+A_d*\frac{\Omega{}p_{k}}{(\Omega\sigma_N^2+\sigma_{\text{ADC}}^2)\|\mathbf{f}_k\|^2}).
\end{equation}
We have that $R_{\text{sum}}(\mathbf{p})$ is concave with respect to the power allocation vector $\mathbf{p}$ and monotonically increasing in each component of the vector $\mathbf{p}$.
\end{lemma}


We define $P_{\text{sum}}\triangleq\sum_{k=1}^{K}p_k/\eta$, which is a linear function on $\mathbf{p}$. Since $R_{\text{sum}}(\mathbf{p})$ is concave on  $\mathbf{p}$, $S_\alpha$ is  convex because $S_\alpha$ is equivalent to $S_\alpha=\{\mathbf{p}\succeq\mathbf{0}|\alpha{}P_c+\alpha{}P_{\text{sum}}-R_{\text{sum}}(\mathbf{p})\leqslant{}0\}$. Therefore,  the strict quasi-concavity of $U(\mathbf{p})$ is proved. 

\section{Proof of Theorem \ref{th:omega-concave}}
\label{apd:2}
For a given power vector $\mathbf{p}$, we can calculate the first derivative of $U(\mathbf{p},\Omega)$ as
\begingroup
\makeatletter\def\f@size{9}\check@mathfonts
\def\maketag@@@#1{\hbox{\m@th\large\normalfont#1}}%
\begin{align}
\notag & \frac{\partial{}U(\mathbf{p},\Omega)}{\partial\Omega} =\\
\notag &\sum_{k=1}^{K}\frac{A_d\sigma_{\text{ADC}}^2{}p_k}{\ln2(P_c+P_{\text{sum}})\left[(\sigma_N^2\Omega+\sigma_{\text{ADC}}^2)\|\mathbf{f}_k\|^2+A_d{}p_k\Omega\right](\sigma_N^2\Omega+\sigma_{\text{ADC}}^2)}\\
\notag &\triangleq{}\sum_{k=1}^{K}\frac{\mathsf{C}}{\mathsf{P}(\Omega)},
\end{align}
\endgroup
where $\mathsf{C}=A_d\sigma_{\text{ADC}}^2{}p_k$ is a positive constant and $\mathsf{P}(\Omega)$ represents a quadratic polynomial of $\Omega$ with all coefficients positive. As a result, we can calculate the second derivative of $U(\mathbf{p},\Omega)$ as
\begin{equation}
\frac{\partial{}^2{}U(\mathbf{p},\Omega)}{\partial\Omega^2}=\sum_{k=1}^{K}\frac{-\mathsf{C}\mathsf{P}'(\Omega)}{\mathsf{P}^2(\Omega)}<0,
\end{equation}
which indicates that $U(\mathbf{p},\Omega)$ is a concave function in $\Omega$ under a given $\mathbf{p}$.

Recall that we use $U(\Omega)$ to denote the maximum energy efficiency under all feasible power vector $\mathbf{p}$. Since the maximum of a set of concave functions is still concave, we have proved the concavity of $U(\Omega)$ on $\Omega$.

\bibliographystyle{IEEEtran}
\bibliography{DAS_LNA2}
\end{document}